\documentclass[12pt,twoside,a4paper]{article}

\usepackage{amsmath}
\usepackage{a4wide}
\usepackage{graphicx}
\usepackage{dsfont}
\usepackage{amsthm}
\usepackage{amssymb}

\newtheorem{theorem}{Theorem}

\newtheorem{remark}{Remark}

\newcommand{\rem}[1]{}

\title{Stabilization of reacting systems}
\author{\"{U}nver \c{C}ift\c{c}i\thanks{\small{uciftci@nku.edu.tr}} $^{1,2}$ \\
\small{$^1$Department of Mathematics}, \small{Nam\i k Kemal University,}\\ \small{59030 Tekirda\u{g}, Turkey}\\[1.5ex]
\small{$^2$Johann Bernoulli Institute for Mathematics and Computer Science,}\\ \small{University of Groningen,} \small{PO Box 407,}\\ 
\small{9700 AK Groningen, The Netherlands} \\
}
\date{}

\begin{document}
\maketitle

\abstract{A feedback stabilization scheme to
stabilize a classical reacting Hamiltonian system
is proposed.
It is based on transforming a saddle-type
equilibrium to an asymptotically stable one, and
is given in a simple and algorithmic way. The question of
destabilization of a stable system to make
a reacting system is also addressed.
The theory is illustrated with  the examples of 
a model Hamiltonian of the form kinetic plus potential, and
the hydrogen atom in crossed and magnetic fields.}

%\noindent
%MSC 2010 numbers: 53C15, 53D17, 70H99
%\noindent 

%\noindent
%Keywords: 
%\noindent 

%%%%%%%%%%%
\section{Introduction}
%\section{\label{}}
Feedback stabilization of nonlinear systems
is a well-established topic in
control theory \cite{Sontag89}. As a special case
the Hamiltonian stabilization appears as a nice
technique due to the rich geometric structure
behind it: Poisson structures. A detailed formulation
of Hamiltonian stabilization can be found in \cite{NijmeijerSchaft90}
with an emphasis to mechanical systems. On the other hand,
a more general treatment which extends the stabilization method
to Poisson manifolds is introduced in \cite{JalnapurkarMarsden00},
and this generality led an application to
systems with symmetry.
For a related work on
stabilization of mechanical systems with symmetry we refer to \cite{Blochetal00,Blochetal01}.
A more recent work \cite{KrechetnikovMarsden07} outlines
controlling dissipation-induced instabilities.

Reaction-type dynamics has had a renewal of understanding
after the development of its phase space geometric picture
\cite{Uzeretal02}. It is based on identifying
geometric structures which govern reaction
dynamics around a saddle-type equilibrium.
Since the introduction of these structures there
have been a big log of work done on
the dynamics of these systems which are
reviewed in \cite{Waalkensetal08}. These
systems not only include chemically reacting systems
but also systems in celestial mechanics, atomic physics, diffusion dynamics in materials,
which have reaction-type dynamics \cite{Waalkensetal08}.
On the other hand
any attempt to stabilize these systems is apparently lacking.

Our aim in this paper is to address the 
problem of stabilization of reacting systems by means of
Hamiltonian stabilization tools. By the special character of
reacting systems, the problem turns out to be making 
a saddle-type equilibrium asymptotically stable
by adding some suitable feedback. The feedback
are algorithmically derived in terms of linearization of
the Hamiltonian vector field. We also touch upon the
problem of the other way around, namely, one can
adopt what is done for saddle equilibria to make
a center-type equilibria saddle-type. Two examples
elucidate the theoretic part of the paper:
a simple reacting system with a model potential, and
the hydrogen atom in crossed and magnetic fields.

We structured the paper
as follows. Sec.~\ref{sec:HS} briefly recalls Hamiltonian
stabilization and Sect.~\ref{sec:gmrs} gives a short introduction
to geometric theory of reactions. We give a detailed explanation of
stabilization of reacting systems in Sec.~\ref{sec:sod},
and a brief look at destabilization of stable systems in Sec.~\ref{sec:ds}.
Examples are given in Sec.~\ref{sec:ex} which are succeeded by conclusions. 
%%%%%%%%%%%
\section{Hamiltonian stabilization}
\label{sec:HS}
We begin with a brief
overview of some basics of the theory of Hamiltonian stabilization
on Poisson manifolds here. A more detailed information can
be found in \cite{NijmeijerSchaft90,JalnapurkarMarsden00}.

Let $P$ be Poisson manifold and
let $H:P \rightarrow \mathbb{R}$
be a Hamiltonian function with the corresponding
Hamiltonian vector field $X_H$. Then the equations
of motion read $\dot{z}=X_H(z)$.
If $z_0$ is an
equilibrium, i.e. $DH(z_0)=0$, the Hessian
$D^2H(z_0)$ is intrinsically defined. Now adding some
inputs $F_i:P \rightarrow \mathbb{R}$,
$i=1,\dots,m$,
such that $F_i(z_0)=0$
by 
\begin{equation}
\dot{z}=X_H(z)-X_{F_1}(z)\,u_1-\dots -X_{F_m}(z)\,u_m 
\end{equation}
gives a system
of which $z_0$ is an equilibrium as well.
Here the feedback are assumed to be $u_i=d_i\,X_{{F}_i}$,
$i=1,\dots,m$, for scalars $d_i>0$. 
Note that, if $F_i$ are in involution, i.e. $\{F_i,F_j\}=0$
for all $i,j=1,\dots,m$, then
$u_i=d_i \, \dot{F}_i$. 

Associated to the closed-loop system
given above the following space is defined:
\begin{equation}
\mathcal{C}=\mbox{span}\{F_i,\{H,F_i\},\{H,\{H,F_i\}\},\dots,\}, 
\end{equation}
$i=1,\dots,m$. 
Here the coefficients of the linear combinations are
real numbers. Accordingly, one defines the
associated co-distribution 
\begin{equation}
 d\mathcal{C}=\mbox{span}\{dg(z)|g\in\mathcal{C}\}.
\end{equation}  
Then we recall the following key result \cite{NijmeijerSchaft90}
for our application. 

\begin{theorem}\label{th:old_theorem}
Let $D^2H({z_0})$ be positive definite
and $\mbox{dim}(d\mathcal{C})=2n$ around $z_0$.
Then the feedback  $u_i=c_i\,X_{{F}_i}$
makes $z_0$ an asymptotically stable equilibrium. 
\end{theorem}

The idea behind the proof is to use LaSalle's Principle
where the Lyapunov function is assumed to be the Hamiltonian,
and the dimensionality condition
makes sure that the only trajectories that lie
in a certain neighborhood are the equilibria \cite{Jalnapurkar99}.
 
\begin{remark}
 A system for which $D^2H(z_0)$ is not positive
can also be stabilized under suitable conditions \cite{NijmeijerSchaft90,JalnapurkarMarsden00}.
But we want to make $D^2H(z_0)$ positive definite
practically by using the linearization below. 
\end{remark}

\begin{remark}
 In general, it is not easy to show whether $d\mathcal{C}(z)$
is constant dimensional or $2n$-dimensional. But if
$F_i$ are independent, $m=n$, and in the form kinetic plus potential,
then  $\mbox{dim}(d\mathcal{C})=2n$ is guarantied \cite{NijmeijerSchaft90}. 
\end{remark}

%%%%%%%%%%%%%%%%
\section{\label{sec:gmrs}Geometry of reacting systems}
In this section, we outline the geometric theory of
reaction dynamics briefly as in the form given in \cite{CiftciWaalkens12}.
A detailed explanation can be found in
\cite{Uzeretal02,Waalkensetal08}, for instance. 

%%%%%
\subsection{The linear case}
\label{sec:phasespacestruc_linear}
Consider the simplest reaction-type Hamiltonian
, i.e. the quadratic Hamiltonian given by
\begin{equation}\label{eq:Hquadratic}
H_2(q,p)= \frac{\lambda}{2} (p_1^2 - q_1^2) + \sum_{k=2}^n  \frac{\omega_k}{2} (p_k^2 + q_k^2) \,,
\end{equation}
where $\lambda,\omega_k>0$. 
Then $DH(0)=0$ and the matrix associated with the linear vector field has
real eigenvalues $\pm \lambda$ and 
complex conjugate imaginary eigenvalues $\pm \mathrm{i}\, \omega_k$, $k=2,\ldots,n$.
Integrability of
the system can be seen by the constants of motion
\begin{equation}\label{eq:def_constants_motion}
\mathcal{I}_1 = p_1^2 - q_1^2\,, \quad   \mathcal{I}_k= p_k^2 + q_k^2 \,,\quad k=2,\ldots,n\,.
\end{equation}
 
Consider a fixed energy $E>0$, where $0$ is the energy of the saddle. Setting
 $q_1=0$ on the energy surface gives the $(2n -2)$-dimensional sphere
\begin{equation}\label{eq:def_DS_linear}
S_{\text{DS}}^{2n -2}  = \{ (q,p) \in \mathbb{R}^{2n } \,:\,  H_2(q,p) = E \,, q_1=0 \}\,.
\end{equation}
The dividing surface $S_{\text{DS}}^{2n -2} $ divides the energy surface into the two components
which have $q_1<0$ (the `reactants') and  $q_1>0$ (the `products'), respectively, and
as $\dot{q}_1= \partial H_2/\partial p_1=\lambda p_1\ne 0$ for $p_1\ne 0$ the dividing surface 
is everywhere transverse to the Hamiltonian flow except for the submanifold where
$q_1=p_1=0$. For $q_1=p_1=0$, one obtains 
$\sum_{k=2}^n  \frac{\omega_k}{2} (p_k^2 + q_k^2) = E$.
The submanifold thus is a  $(2n-3)$-dimensional sphere which we denote by
\begin{equation}\label{eq:NHIM_linear}
S_{\text{NHIM}}^{2n-3}  = \{ (q,p) \in \mathbb{R}^{2n} \,:\,  
H_2(q,p) = E \,, q_1=p_1=0  \}\,.
\end{equation}
This is a so called \emph{normally hyperbolic invariant manifold} \cite{Wiggins94} (NHIM for short),
i.e. $S_{\text{NHIM}}^{2n-3}$ is invariant (since $q_1=p_1=0$ implies $\dot{q}_1=\dot{p}_1=0$)
and the contraction and expansion rates for motions on $S_{\text{NHIM}}^{2n-3}$ are dominated by
those components related to directions transverse to $S_{\text{NHIM}}^{2n -3}$. 
The NHIM \eqref{eq:NHIM_linear} can be considered to form the equator of the dividing surface
\eqref{eq:def_DS_linear} in the sense that it divides it into two hemispheres which topologically are
$(2n-2)$-dimensional balls. All forward reactive trajectories (i.e. trajectories moving from reactants
to products) cross one of these hemispheres, and all backward reactive trajectories (i.e. trajectories moving
from products to reactants) cross the other of these hemispheres. 
Note that a trajectory is reactive only if it has $\mathcal{I}_1>0$ (i.e. if it has sufficient energy in the
first degree of freedom). Trajectories with  $\mathcal{I}_1<0$ are nonreactive, i.e. they stay on the side of
reactants or on the side of products. See Fig.~\ref{fig:saddle_center}
for the phase portrait of the system. A 
forward reactive trajectory is depicted by the dashed curves.
 
\begin{figure*}
\begin{center}
\raisebox{7cm}{(a)}\includegraphics[angle=0,width=6cm]{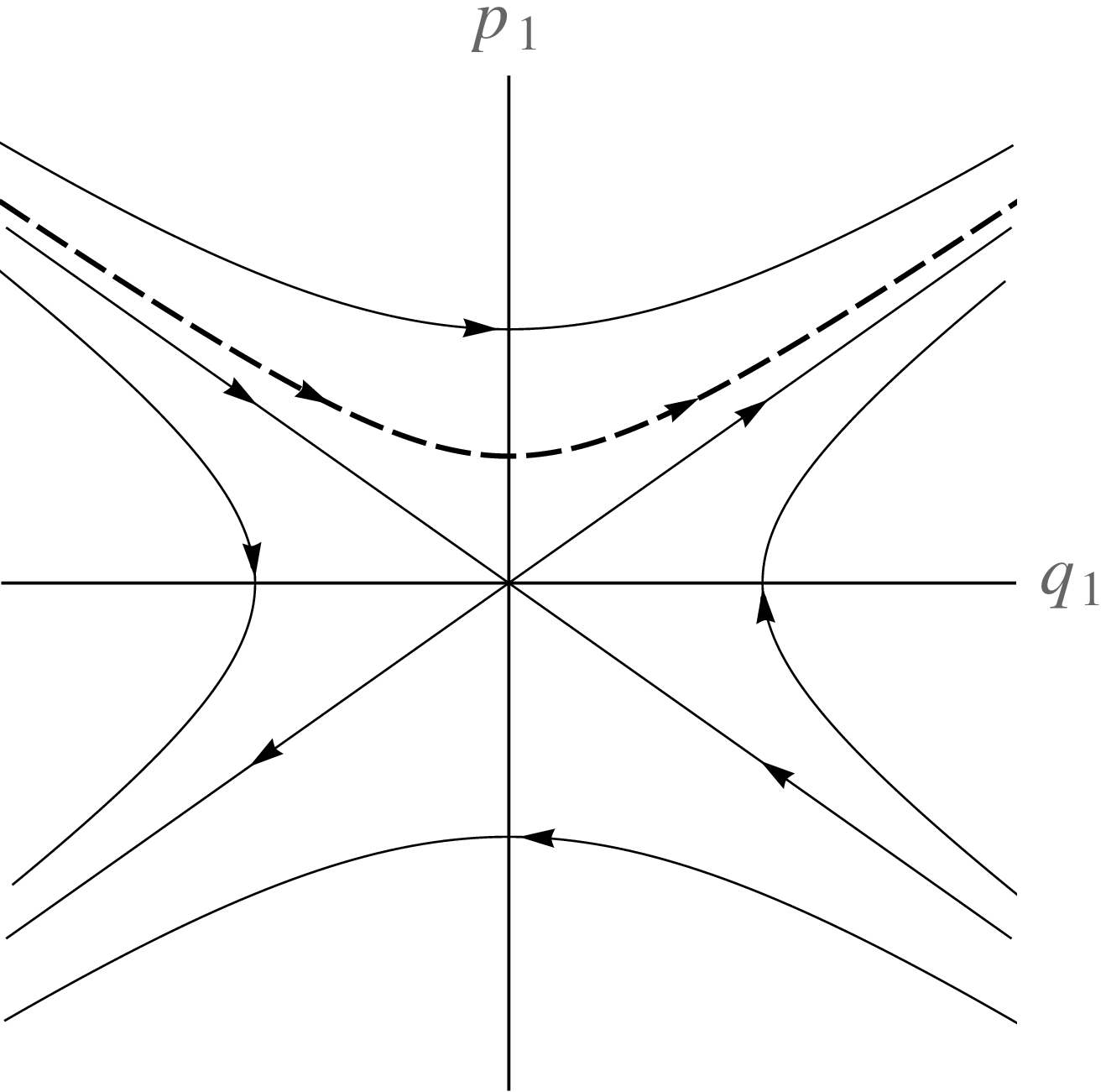} 
\raisebox{7cm}{(b)}\includegraphics[angle=0,width=6cm]{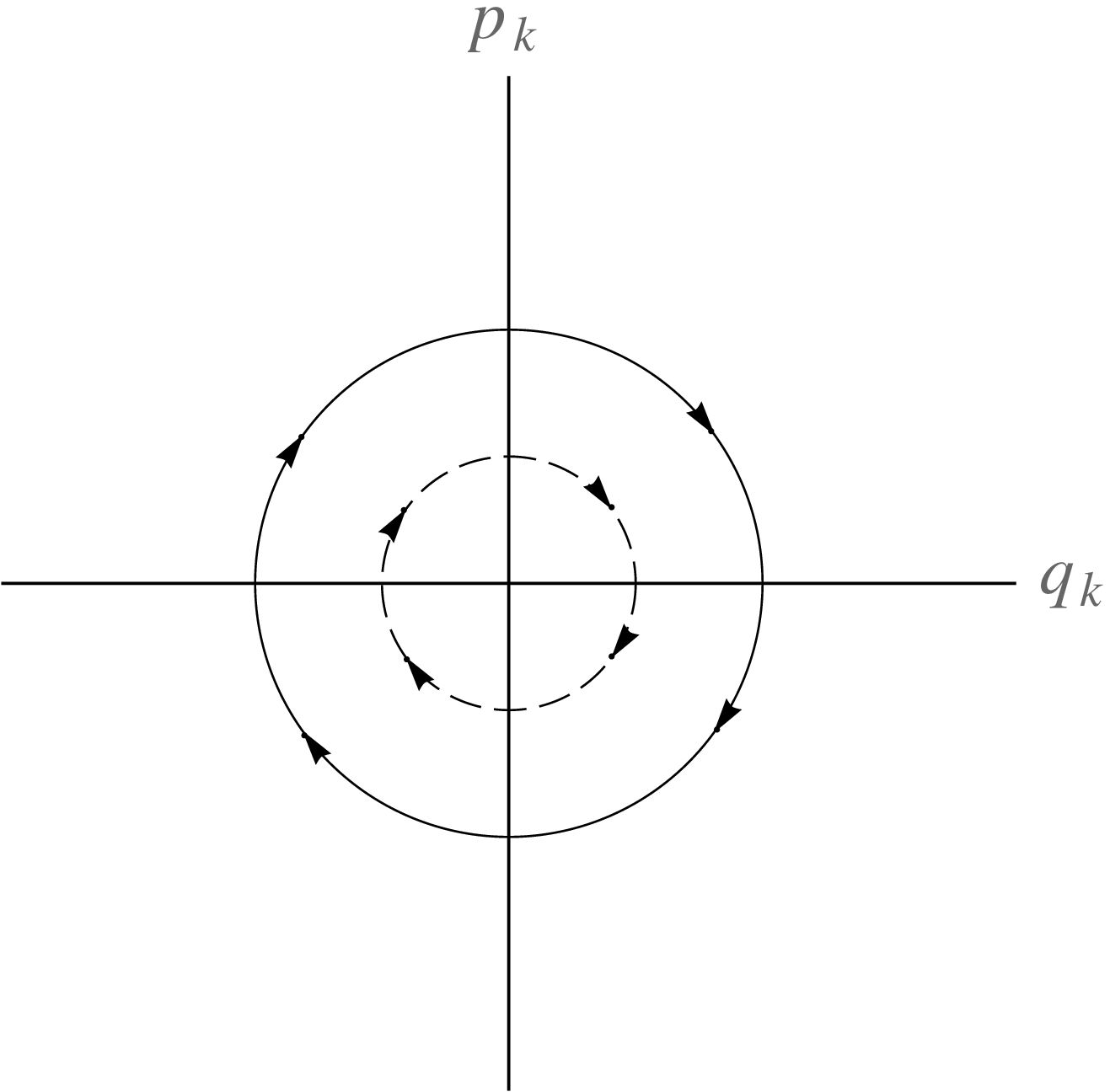}
\end{center}
\caption{\label{fig:saddle_center}
Phase portrait of the phase space. Projections of a reactive trajectory
(dashed curves) into
(a) the saddle plane, (b) the center planes.
}
\end{figure*}

%%%%%%
\subsection{The general (nonlinear) case}
\label{sec:phasespacestruc_nonlinear} 
Consider a Hamiltonian $H:\mathbb{R}^{2n} \rightarrow \mathbb{R}$ with
an equilibrium at the origin $(q,p)=(0,0)$ for some canonical coordinates
$(q,p)$. Assume that $H$ has a saddle-center-\dots-center stability type equilibrium,
that is, the matrix associated with
the linearization at $(0,0)$ of the Hamiltonian vector field
has eigenvalues $\mp \lambda$, $\lambda>0$,  and $\mp \mathrm{i} \,\omega_k$, $\omega_k>0$,
$k=1,\dots,n-1 $.   
Further assume that the submatrix corresponding to
the imaginary eigenvalues is semisimple.
In the neighborhood of the saddle the dynamics is thus similar to that of the
linear vector field described in Sec.~\ref{sec:phasespacestruc_linear}. 
In fact if follows from general principles that all the phase structures discussed
in Sec.~\ref{sec:phasespacestruc_linear} persist in the neighborhood of the saddle
(which in particular implies that one has to restrict to energies close to the energy of the saddle).  
Moreover, these phase space structures can be constructed in an algorithmic fashion  using a
Poincar{\'e}-Birkhoff normal form \cite{Uzeretal02,Waalkensetal08}.
Assuming that the eigenvalues $\omega_k$, $k=2,\ldots,n$, are independent over the
field of rational numbers (i.e. in the absence of resonances), 
the Poincar{\'e}-Birkhoff normal form yields a symplectic transformation to new
(\emph{normal form}) coordinates such that 
the transformed Hamiltonian function truncated at order $n_0$ of its Taylor expansion assumes the  form 
\begin{equation} 
H_{\text{NF}}({\cal I}_1,\mathcal{I}_2,\ldots,\mathcal{I}_n)\,,
\end{equation}
where ${\cal I}_1$ and $I_k$, $k=2,\ldots,n$, are constants of motions which
(when expressed in terms of the normal form coordinates) have the same form as in
\eqref{eq:def_constants_motion}, and $H_{\text{NF}}$ is a polynomial of order
$n_0/2$ in  ${\cal I}_1$ and $I_k$, $k=2,\ldots,n$ (note that only even orders $n_0$
of a normal form make sense).

In terms of  the normal form coordinates the phase space structures can be defined in
a manner which is virtually identical to the linear case by replacing
$H_2(q,p)$ by $H_{\mathrm{NF}}({\cal I}_1,I_2,\ldots,\mathcal{I}_n)$ in the definitions
in Sec.~\ref{sec:phasespacestruc_linear}.  Using then the inverse of the normal form transformation
allows one to construct the phase space structures in the original (`physical') coordinates.
As it is seen
the Poincar{\'e}-Birkhoff normal form is the main tool
in defining phase space structures, but
we only review the first order linearization in the next section, as
it serves enough for our purpose of stabilization.

%%%%%%%%%%%%%%%%%%
\section{Stabilization of saddle-type equilibria}
\label{sec:sod}
In this section we consider a nonlinear Hamiltonian system of
reaction-type, and we give a result which gives 
an algorithmic way of making the system
asymptotically stable around the given
equilibrium. This is done as follows.

Let $P$ be a Poisson manifold, 
$(z_1,\dots,z_{2n})=(x_1,\dots,x_n,P_1,\dots,P_n)$ be
a canonical coordinate system around $z_0 \in P$
which is set to be the origin $(0,\dots,0)$, 
and $H:P\rightarrow \mathbb{R}$
be an analytic Hamiltonian function with
the corresponding Hamiltonian vector field $X_H$
having $z_0$ as an equilibrium point
of type saddle-center-\dots-center.
So we assume that
the linearization matrix of $X_H$, or in other words the matrix
$J \, D^2H$ evaluated at $z_0$,  has eigenvalues 
$\mp \lambda, \, \mp \mathrm{i} \,
\omega_k$, 
for reals $\lambda,\omega_k>0$, $k=2,\dots,n$. 

One can put the quadratic part of
$H$ into the form (\ref{eq:Hquadratic}) 
by a symplectic change of coordinates \cite{Waalkensetal08}.
To do this label the eigenvalues
by 
\begin{equation}
e_1=\lambda=-e_{n+1},\quad e_k=\mathrm{i}\,\omega_k=-e_{k+n},\quad k=2,\dots,n,
\end{equation}
and corresponding eigenvalues by
$v_1,\dots,v_{2n}$. Consider the following symplectic
matrix 
\begin{equation}
 M=(c_1\,v_1,c_2\,\mbox{Re}v_2,\dots,c_n\,\mbox{Re}v_n,c_1\,v_{1+n},c_2\,\mbox{Im}v_2,\dots,c_n\,\mbox{Im}v_n)
\end{equation}
where 
\begin{equation}
 c_1=\langle v_1,J\,v_{1+n}\rangle^{-1/2},\quad  c_k=\langle \mbox{Re}v_k,J\,\mbox{Im}v_{k}\rangle^{-1/2},
\quad k=2,\dots,n.
\end{equation}
Then the coordinate transformation
\begin{equation}
M \cdot
\left[ 
\begin{array}{c}
   x\\
  P
 \end{array}
\right]
=
\left[ 
\begin{array}{c}
   \bar{q}\\
 \bar{p}
 \end{array}
\right]
\end{equation}
gives a new canonical coordinate 
system $(\bar{q}_1,\dots,\bar{q}_n,\bar{p}_1,\dots\,\bar{p}_n)$,
and in these coordinates
the quadratic part of the Hamiltonian takes the form
\begin{equation}\label{eq:Hquadratic_1}
H_2(\bar{q},\bar{p})= \lambda \, \bar{q}_1 \,\bar{p}_1 + \sum_{k=2}^n  \frac{\omega_k}{2} (\bar{p}_k^2 
+\bar{q}_k^2) \,.
\end{equation}
Let a rotation $N$ be given by
\begin{equation}
\begin{split}
q_1 &= \frac{1}{\sqrt{2}}\, \bar{q}_1- \frac{1}{\sqrt{2}}\, \bar{p}_1, \quad 
p_1 = \frac{1}{\sqrt{2}}\, \bar{q}_1+\frac{1}{\sqrt{2}}\, \bar{p}_1,  \\
q_k &=\bar{q}_k, \quad p_k=\bar{p}_k, \quad k=2,\dots n,
\end{split}
\end{equation}
then if we introduce
\begin{equation}\label{eq:matrix_s}
S=N \cdot M^{-1}
\end{equation} 
the coordinate transformation
\begin{equation}
S \cdot
\left[ 
\begin{array}{c}
   x\\
  P
 \end{array}
\right]
=
\left[ 
\begin{array}{c}
   q\\
 p
 \end{array}
\right]
\end{equation}
gives also a set of canonical
coordinates $(q_1\dots,q_n,p_1\dots,p_n)$. Then in
$(q,p)$ coordinates $H_2$
assumes the form (\ref{eq:Hquadratic}).

Consider the controls 
\begin{equation}
u_i=-c\,F_1-d_i\,\dot{F}_i, \quad i=1,\dots,n,
\end{equation}
with any constants such that $c> \lambda$ and $d_i>0, \, i=1,\dots,n$,
where the functions
$F_i:P\rightarrow \mathbb{R}$
are given by
\begin{equation}
F_i=\sum_{j=1}^{n}S_{ij}\,z_j
\end{equation}
for the matrix $S$ introduced in Eq.~\ref{eq:matrix_s}.
Then we prove
\begin{theorem}\label{main_th}
With the notion above, the system
\begin{equation}
\dot{z}=X_H(z)-X_{F_1}(z)\,u_1-\dots -X_{F_n}(z)\,u_n 
\end{equation}
is asymptotically stable
around $z_0$. 
\end{theorem}
\begin{proof}
We want to show that the conditions
of Theorem~\ref{th:old_theorem}
are satisfied. 
First, it is easily seen that
$F_i(z_0)=\sum_{j=1}^{n}S_{ij}\,z_j(z_0)=0$ for
$i=1,\dots,n$.
Then we add the control $v=-c\,F_1$
and consider the new system
\begin{equation}\label{eq:first_system}
\dot{z}=X_H-X_{F_1}(z)\,v.
\end{equation}
One can check that the system (\ref{eq:first_system})
is also Hamiltonian with the
modified Hamiltonian $H_{\mbox{mod}}=H+\frac{1}{2}\,c\,F_1^2$
and the Hessian matrix $D^2H(z_0)$ is positive definite.
In fact,
\begin{equation}
\dot{z}=X_H-X_{F_1}(z)\,v=X_H+c\,X_{F_1}(z)\,F_1
=X_H+c\,X_{\frac{1}{2}\,F_1^2}(z)
=X_{H+\frac{1}{2}\,c\,F_1^2}(z),
\end{equation}
and in normal form coordinates
$F_i=q_i$, $i=1,\dots,n$, so in these coordinates
\begin{equation}
D^2H_{\mbox{mod}}(q,p)=D^2H(q,p)+c\,E_{11},
\end{equation}
where $E_{11}$ is the matrix with
zero entries  except the first entry
equal to $1$. This shows that
\begin{equation}
D^2H_{\mbox{mod}}(z_0)=
\mbox{diag}(-\lambda+c,\lambda,\omega_1,\omega_1,\dots,\omega_n,\omega_n).
\end{equation}
which is, clearly, positive definite since $c>\lambda$. 

So far, we have ensured the positive definiteness condition in   
Theorem~\ref{th:old_theorem}. Next we add the controls $v_i=-d_i\,\dot{F}_i$, $i=1,\dots,n$,
and it remains to show that the
functions $F_i$, 
$i=1,\dots,n$
satisfy the dimensionality assumption
$\mbox{dim}(d\mathcal{C})=2n$
for 
\begin{equation}
d\mathcal{C}(z)=\mbox{span}\{dg(z)|g\in \mathcal{C}\},
\end{equation}
where
\begin{equation}
\mathcal{C}=\mbox{span}\{F_i,\{H_{\mbox{mod}},F_i\},\{H,\{F_i\}\},\dots \}, \quad i=1,\dots n,
\end{equation}
and $F_i$ are in involution.
To see this
it is observed that
\begin{equation}
\mathcal{C}=\mbox{span}\{q_i,\{H_{\mbox{mod}} \circ S^{-1},q_i\},\{H,\{q_i\}\},\dots \}, \quad i=1,\dots n.
\end{equation}
In fact, by $F_i=q_i \circ S$ we have
\begin{equation}
\{H_{\mbox{mod}},F_i \}=\{H_{\mbox{mod}},q_i \circ S\}=\{H_{\mbox{mod}} \circ S^{-1}, q_i \}
\end{equation}
since $S$ is a Poisson map. But $H_{\mbox{mod}} \circ S^{-1}$ is real analytic
so we can write it as a Taylor series around $z_0$ where
the quadratic part is given by 
(\ref{eq:Hquadratic}). Then it can be seen that
$f_i=\{H_{\mbox{mod}} \circ S^{-1},q_i\},\, i=1,\dots,n$,
are independent, because one has
\begin{equation}
\begin{split}
 f_1&=\{H_{\mbox{mod}} \circ S^{-1},q_1\}=
\lambda \, p_1 + \mbox{h.o.t.}, \\
f_k&=\{H_{\mbox{mod}} \circ S^{-1},q_k\}=
\omega_k \, p_k + \mbox{h.o.t.}, \quad k=1,\dots,n.
\end{split}
\end{equation}
Furthermore, 
$\{q_1,\dots,q_n,f_1,\dots,f_n\}$ forms a set of $2n$
independent functions. So, $d\mathcal{C}(z)$ is $2n$-dimensional, in particular
\begin{equation}
d\mathcal{C}(z_0)=\mbox{span}\{dq_1(z_0),\dots,dq_n(z_0),dp_1(z_0),\dots,dp_n(z_0)\}.
\end{equation} 
 
 As the final step we need to show that
$F_i$ are in involution. This can also be seen
easily by 
\begin{equation}
\{F_i,F_j \}=\{q_i \circ S,q_j\circ S\}=\{q_i ,q_j\}=0.
\end{equation}
\end{proof}

The feedback added system is no more conservative because of the dissipative
inputs $v_i$. As the system is asymptotically
stable, trajectories projected into the phase planes look like the ones in
Fig.~\ref{fig:asymptotic}.   

\begin{figure*}
\begin{center}
\raisebox{7cm}{}\includegraphics[angle=0,width=6cm]{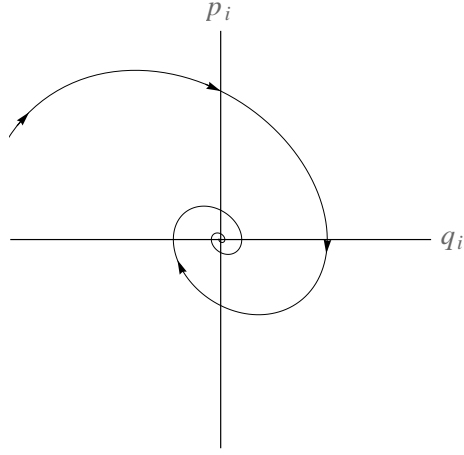}
\end{center}
\caption{\label{fig:asymptotic}
Projetion of trajectories near the equilibrium into the phase planes.}
\end{figure*}

%%%%%%%%%%%%%%%%%%%%%%%
\section{Destabilization of a stable system to make a reacting one}
\label{sec:ds}
A similar procedure as in Sec.~\ref{sec:sod} can also be applied to
a stable system with purely complex eigenvalues in order to
obtain an unstable system with a saddle.

We cansider again a Poisson manifold
$P$, a canonical coordinate system denoted by
$(z_1,\dots,z_{2n})=(x_1,\dots,x_n,P_1,\dots,P_n)$
around $z_0 \in P$
which is set to be the origin $(0,\dots,0)$, 
and be an analytic Hamiltonian $H:P\rightarrow \mathbb{R}$ with
the corresponding Hamiltonian vector field $X_H$
having $z_0$ as an equilibrium point
of type center-\dots-center.
So we assume that
the linearization matrix of $X_H$, or in other words the matrix
$J \, D^2H$ evaluated at $z_0$,  has eigenvalues 
$ \mp \mathrm{i} \, \omega_i$, 
for reals $\omega_i>0$, $i=1,\dots,n$. 

The quadratic part of the Hamiltonian can be put in the form
$H$ into the form 
\begin{equation}\label{eq:Hquadratic_center}
H_2(q,p)= \sum_{i=1}^n  \frac{\omega_i}{2} (p_i^2 + q_i^2) \,,
\end{equation}
by a symplectic change of coordinates \cite{Waalkensetal08}
as recalled in the following.
Label the eigenvalues by 
\begin{equation}
e_i=\mathrm{i}\,\omega_i=-e_{i+n},\quad i=1,\dots,n,
\end{equation}
and corresponding eigenvalues by
$v_1,\dots,v_{2n}$. Consider the following symplectic
matrix 
\begin{equation}
 M=(c_1\,\mbox{Re}v_1,\dots,c_n\,\mbox{Re}v_n,c_1\,\mbox{Im}v_1,\dots,c_n\,\mbox{Im}v_n)
\end{equation}
where 
\begin{equation}
 c_i=\langle \mbox{Re}v_i,J\,\mbox{Im}v_{i}\rangle^{-1/2},
\quad i=1,\dots,n.
\end{equation}
Then the coordinate transformation
\begin{equation}\label{eq:matrix_m_n}
M \cdot
\left[ 
\begin{array}{c}
   x\\
  P
 \end{array}
\right]
=
\left[ 
\begin{array}{c}
   q\\
   p
 \end{array}
\right]
\end{equation}
gives a new canonical coordinate 
system $(q_1,\dots,q_n,p_1,\dots\,p_n)$,
and in these coordinates
the quadratic part of the Hamiltonian takes the form
(\ref{eq:Hquadratic_center}).

Consider the control
\begin{equation}
u=-\,c\,F_1
\end{equation}
with any constants such that $c> \omega_1$ the function
$F_1:P\rightarrow \mathbb{R}$
is given by
\begin{equation}
F_1=\sum_{j=1}^{n}M_{1j}\,z_j
\end{equation}
for the matrix $M$ introduced in Eq.~\ref{eq:matrix_m_n}.
Then we prove
\begin{theorem}
With the notion above, the system
\begin{equation}
\dot{z}=X_H(z)-X_{F_1}(z)\,u_1
\end{equation}
is of type $saddle-center-\dots-center$.
\end{theorem}
\begin{proof}
Similar to the proof of Theorem~\ref{main_th}.
\end{proof}

%%%%%%%%%%%%%%%%%%%%%
\section{Examples}
\label{sec:ex}
 We illustrate the procedure of making
a saddle-type equilibrium asymptotically
stable with two examples.

\subsection{A model example}
The following system represents a typical pattern of
isomerization reactions  \cite{TachibanaFukui78}.

Consider a system with potential 
function
\begin{equation}\label{potential}
V(x_1,x_2)=\frac{1}{a^2} \, x_1^2 \, (x_1 - 1)^2 + \frac{1}{b^2}\, x_2^2, \quad a>b,
\end{equation}
and 
Hamiltonian
\begin{equation}
H(x_1,x_2,P_1,P_2)=\frac{1}{2}\,P_1^2+\frac{1}{2}\,P_2^2
+V(x_1,x_2).
\end{equation}
Clearly, $H$ has three equilibria
which are critical points of $V$. 
These are two centers; $(x_1,x_2)=(0,0)$ and
 $(x_1,x_2)=(1,0)$, and one saddle; $(x_1,x_2)=(1/2,0)$.
As we are interested in the saddle,
we translate the coordinates by
$(x_1,x_2)\mapsto (x_1+1/2,x_2)$
to move the saddle to the origin. 
We use the same notation for the translated
coordinates and the potential, then we have
\begin{equation}\label{potential_new}
V(x_1,x_2)=\frac{1}{a^2} \, x_1^2 \, (x_1 - 1)^2 + \frac{1}{b^2}\, x_2^2, \quad a>b
\end{equation}
which will be used henceforth. The contours of the potential surface 
are depicted in Fig.~\ref{fig:saddle_equipotential}. 

\begin{figure*}
\begin{center}
\raisebox{7cm}{(a)}\includegraphics[angle=0,width=6cm]{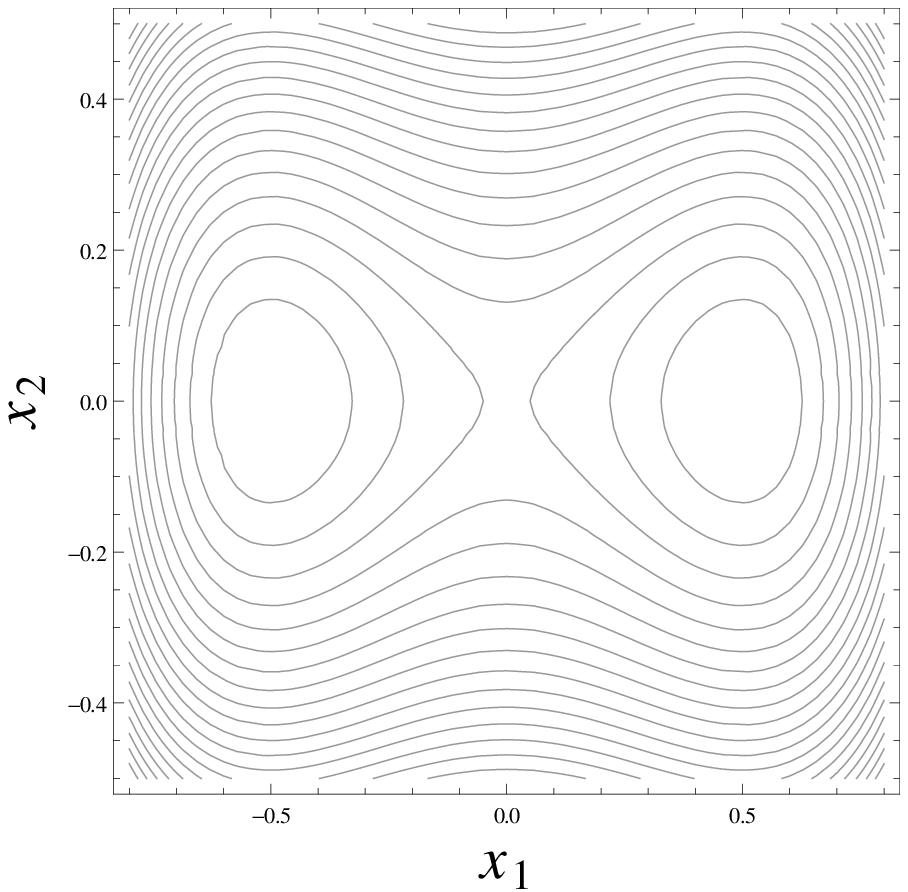} 
\raisebox{7cm}{(b)}\includegraphics[angle=0,width=6cm]{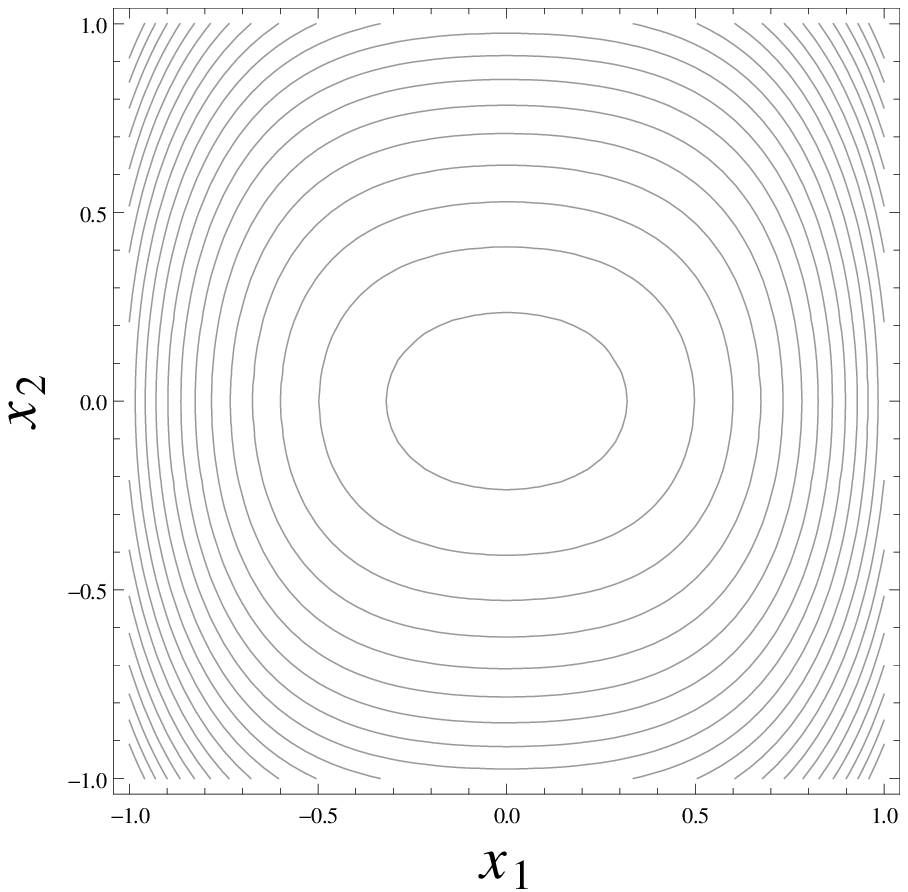} 
\end{center}
\caption{\label{fig:saddle_equipotential}
Contours of potential functions; (a) the original model
potential (\ref{potential_new}), (b) the modified potential
(\ref{mod_pot}). }
\end{figure*}

The Hessian matrix $J \, D^2H$ at point $z_0=(0,0,0,0)$ is computed
to be
\begin{equation}
J \, D^2H(z_0)=
\left[
\begin{array}{cccccc}
0&0&1&0 \\
0&0&0&1\\
\frac{1}{a^2}&0&0&0 \\
0&-\frac{2}{b^2}&0&0\\
\end{array}
\right]
\end{equation}
which has eigenvalues
$e_1=\frac{1}{a}=-e_3, \quad e_2= \frac{\sqrt{2}}{b}\,\mathrm{i}=-e_4$
with corresponding eigenvectors
$v_1=\{a, 0, 1, 0\}, \, v_3=\{-a, 0, 1, 0\} , \, v_2=\{0,- \frac{b}{\sqrt{2}}\,\mathrm{i}, 0, 1\},
 \, v_4=\{0,- \frac{b}{\sqrt{2}}\,\mathrm{i}, 0, 1\}$.
So, we have
$c_1 = (v_1\cdot J \, v_3)^{-1/2}=(2\,a)^{-1/2}, \,
c_2 = (\mbox{Re}(v_2) \cdot J  \, \mbox{Im}(v_2))^{-1/2}
=2^{1/4} \, b^{-1/2}.$
Then the matrix $M$ reads
\begin{equation}
M=
\left[
\begin{array}{cccccc}
\frac{\sqrt{a}}{\sqrt{2}}&0&-\frac{\sqrt{a}}{\sqrt{2}}&0 \\
0&0&0&-\frac{\sqrt{b}}{2^{1/4}} \\
\frac{1}{\sqrt{2\,a}}&0&\frac{1}{\sqrt{2\,a}}&0 \\
0&\frac{2^{1/4}}{\sqrt{b}}&0&0\\
\end{array}
\right].
\end{equation}
Finally, the matrix $S$
which is the multiplication of
the rotation matrix 
\begin{equation}
R=
\left[
\begin{array}{cccccc}
\frac{1}{\sqrt{2}}&0&-\frac{1}{\sqrt{2}}&0 \\
0&0&0&0  \\
\frac{1}{\sqrt{2}}&0&\frac{1}{\sqrt{2}}&0 \\
0&0&0&0\\
\end{array}
\right]
\end{equation}
and the matrix
$M$ becomes
\begin{equation}
S=
\left[
\begin{array}{cccccc}
\frac{1}{\sqrt{a}}&0&0&0 \\
0&0&0&\frac{\sqrt{b}}{2^{1/4}}  \\
0&0&\sqrt{a}&0 \\
0&-\frac{2^{1/4}}{\sqrt{b}}&0&0\\
\end{array}
\right].
\end{equation}
Hence the functions $F_1, F_2$  
are obtained to be $F_1=\frac{x_1}{\sqrt{a}},\, F_2=\frac{\sqrt{b}}{2^{1/4}} \, P_2$.
Observe that, the modified Hamiltonian $H_{\mbox{mod}}$ has the form
$H_{\mbox{mod}}=H+\frac{1}{2 \, a^2} \, x_1^2$, and
the modified potential is 
\begin{equation}\label{mod_pot}
 V_{\mbox{mod}}=V+\frac{1}{2 \, a^2} \, x_1^2
\end{equation}
of which contours are depicted in Fig.~\ref{fig:saddle_equipotential} (b).
After the addition of associated controls, we have the system with equations of motion
\begin{equation}
\begin{split} 
\dot{x}_1&=P_1, \\
\dot{x}_2&=P_2+\frac{\sqrt{2}\,x_2}{b}, \\
\dot{P}_1&=-\frac{4\,x_1^3}{a^2}+\frac{P_1}{a}, \\
\dot{P}_2&=-\frac{2\,x_2}{b^2}, 
\end{split}
\end{equation}
where we choose $d_1=d_2=1$.

%%%
\subsection{Hydrogen atom in crossed and magnetic fields}

The following example is a Hamiltonian system which is
not of the form kinetic plus potential. We do not give the original
form but a form obtained after some manipulations \cite{Uzeretal02}.

The Hamiltonian can be put in the form
\begin{equation}
H=\frac{1}{2}\,(P_1^2
+P_2^2+P_3^2)-\frac{1}{R}+\frac{1}{2}\,(x_1\,P_2-x_2\,P_1)+\frac{1}{8}\,
(x_1^2+x_2^2)-\epsilon \, x_1,
\end{equation}
where
$R=(x1^2+x_2^2+x_3^2)^{1/2}$.
We will consider the experimentally interesting value $\epsilon=0.58$ 
henceforth. The Stark saddle point in atomic physics corresponds to the
point $(x_1,x_2,x_3,P_1,P_2,P_3)=(\epsilon^{1/2},0,0,0,-\frac{1}{2}\,\epsilon^{-1/2})$.
So after a coordinate shift 
$(x_1,x_2,x_3,P_1,P_2,P_3)\mapsto (x_1-\epsilon^{1/2},x_2,x_3,P_1,P_2,P_3
+\frac{1}{2}\,\epsilon^{-1/2}),$
by retaining the same notation for 
the translated coordinates and the new Hamiltonian we have
\begin{equation}
H=\frac{1}{2}\,(P_1^2
+P_2^2+P_3^2)-\frac{1}{R}+\frac{1}{2}\,(x_1\,P_2-x_2\,P_1)+\frac{1}{8}\,
(x_1^2+x_2^2)-\epsilon \, x_1-\epsilon^{1/2},
\end{equation}
where $R=\left[(x_1+\epsilon^{1/2})^2+x_2^2+x_3^2\right]^{1/2}$.
Then the matrix $J \, D^2H(0)$ is obtained to be
\begin{equation}
J \, D^2H(0)=
\left[
\begin{array}{cccccc}
0&-0.5&0&1&0&0 \\
0.5&0&0&0&1&0\\
0&0&0&0&0&1 \\
0.63343&0&0&0&-0.5&0\\
0&-0.691715&0&0.5&0&0\\
0&0&-0.441715&0&0&0
\end{array}
\right]
\end{equation}
which has eigenvalues
\begin{equation}
e_1=0.63645=-e_4, \quad e_2=0.981506 \, \mathrm{i}=-e_5, \quad e_3=0.664616 \,\mathrm{i}=-e_6
\end{equation}
with corresponding eigenvectors
\begin{equation}
 \begin{split}
v_{1,4}&=\{\pm 0.63645, 0.478361, 0, 0.644249, \mp 0.0137719, 0\}, \\
v_{2,5}&=\{\pm 0.981506 \,\mathrm{i}, 1.84678, 0, -0.0399619, 
\pm 1.32188 \, \mathrm{i}, 0\} , \\
v_{3,6}&=\{0, 0,  \mp 1.50463 \,\mathrm{i}, 0, 0, 1\}.
\end{split}
\end{equation}
So, we have
\begin{equation}
 \begin{split}
c_1 &= (v_1\cdot J \, v_4)^{-1/2}=1.09551, \\
c_2 &= (\mbox{Re}(v_2) \cdot J  \, \mbox{Im}(v_2))^{-1/2}
= 0.634944 \\
c_3 &= (\mbox{Re}(v_3) \cdot J  \, \mbox{Im}(v_3))^{-1/2}=0.81524.
\end{split}
\end{equation}
Then the matrix $M$ reads
\begin{equation}
M=
\left[
\begin{array}{cccccc}
 0.697235 & 0 & 0 & -0.697235 & 0.623201 & 0 \\
 0.524048 & 1.1726 & 0 & 0.524048 & 0 & 0 \\
 0 & 0 & 0 & 0 & 0 & -1.22663 \\
 0.705779 & -0.0253736 & 0 & 0.705779 & 0 & 0 \\
 -0.0150872 & 0 & 0 & 0.0150872 & 0.839317 & 0 \\
 0 & 0 & 0.81524 & 0 & 0 & 0 \\
\end{array}
\right].
\end{equation}
Finally, the matrix $S$
which is the multiplication of
the rotation matrix $R$ and the matrix
$M$ becomes
\begin{equation}
S=
\left[
\begin{array}{cccccc}
 0.998122 & 0 & 0 & 0 & -0.741116 & 0 \\
 0 & 0.839317 & 0 & -0.623201 & 0 & 0 \\
 0 & 0 & 0 & 0 & 0 & 1.22663 \\
 0 & 0.0213366 & 0 & 0.986039 & 0 & 0 \\
 0.0253736 & 0 & 0 & 0 & 1.1726 & 0 \\
 0 & 0 & -0.81524 & 0 & 0 & 0 \\
\end{array}
\right].
\end{equation}
Hence the functions $F_1,F_2,F_3$ are
derived as
\begin{equation}
 \begin{split}
F_1&=-0.741116 \,P_2 + 0.998122 \,x_1, \\
F_2&=-0.623201 \,P1 + 0.839317 \,x_2,\\
F_3&=1.22663 \,P_3.
\end{split}
\end{equation}
This way, the new system is made asymptotically stable around the origin.

%%%%%%%%%%%%%%%%%%%%%%%%%%
\section{Conclusions and future work}
An algorithmic stabilization of reacting
systems is outlined. It relays on the linearization
of the Hamiltonian vector field around
the equilibrium. The examples reflect the novelty of
the technique given in the paper.
Next step is to do a study for Hamiltonian
systems with symmetry where the equilibria are replaced by relative equilibria.
This can be done by using
canonical coordinates on the reduced space instead of 
the reduced energy momentum method as in \cite{JalnapurkarMarsden00}. A derivation
method of canonical coordinates on a
reduced space for $N$-body reduction
is outlined in \cite{CiftciWaalkens12}
and for cotangent bundle reduction is given in \cite{Ciftcietal}.

 %%%%%%%%%%%%%%
\bibliographystyle{unsrt}
\bibliography{srcbib}

\begin{thebibliography}{10}

\bibitem{Sontag89}
E.~D. Sontag.
\newblock Feedback stabilization of nonlinear systems.
\newblock In {\em Mathematical Theory of Networks and Systems. Birkhauser},
  pages 61--81. Birkhauser, 1989.

\bibitem{NijmeijerSchaft90}
H.~Nijmeijer and A.~van~der Schaft.
\newblock {\em Nonlinear dynamical control systems}.
\newblock Springer-Verlag, New York, 1990.

\bibitem{JalnapurkarMarsden00}
S.~M. Jalnapurkar and J.~E. Marsden.
\newblock Stabilization of relative equilibria.
\newblock {\em IEEE Trans. Automat. Control}, 45(8):1483--1491, 2000.

\bibitem{Blochetal00}
A.~M. Bloch, N.~E. L.eonard, and J.~E. Marsden.
\newblock Controlled {L}agrangians and the stabilization of mechanical systems.
  {I}. {T}he first matching theorem.
\newblock {\em IEEE Trans. Automat. Control}, 45(12):2253--2270, 2000.

\bibitem{Blochetal01}
A.~M. Bloch, D.~E. Chang, N.~E. Leonard, and J.~E. Marsden.
\newblock Controlled {L}agrangians and the stabilization of mechanical systems.
  {II}. {P}otential shaping.
\newblock {\em IEEE Trans. Automat. Control}, 46(10):1556--1571, 2001.

\bibitem{KrechetnikovMarsden07}
R.~Krechetnikov and J.~E. Marsden.
\newblock Dissipation-induced instabilities in finite dimensions.
\newblock {\em Rev. Mod. Phys.}, 79:519--553, Apr 2007.

\bibitem{Uzeretal02}
T.~Uzer, C.~Jaff{\'e}, J.~Palaci{\'a}n, P.~Yanguas, and S.~Wiggins.
\newblock The geometry of reaction dynamics.
\newblock {\em Nonlinearity}, 15:957--992, 2002.

\bibitem{Waalkensetal08}
H.~Waalkens, R.~Schubert, and S.~Wiggins.
\newblock Wigner's dynamical transition state theory in phase space: classical
  and quantum.
\newblock {\em Nonlinearity}, 21(1):R1--R118, 2008.

\bibitem{Jalnapurkar99}
S.~M. Jalnapurkar.
\newblock {\em Modeling and stabilization for mechanical systems}.
\newblock ProQuest LLC, Ann Arbor, MI, 1999.
\newblock Thesis (Ph.D.)--University of California, Berkeley.

\bibitem{CiftciWaalkens12}
\"{U}. \c{C}ift\c{c}i and H.~{Waalkens}.
\newblock {Phase space structures governing reaction dynamics in rotating
  molecules}.
\newblock {\em Nonlinearity}, 25:791--892, 2012.

\bibitem{Wiggins94}
S.~Wiggins.
\newblock {\em Normally {H}yperbolic {I}nvariant {M}anifolds in {D}ynamical
  {S}ystems}.
\newblock Springer, Berlin, 1994.

\bibitem{TachibanaFukui78}
A.~Tachibana and K.~Fukui.
\newblock Differential geometry of chemically reacting systems.
\newblock {\em Theoretica chimica acta}, 49:321--347, 1978.

\bibitem{Ciftcietal}
\"{U}. \c{C}ift\c{c}i, H.~{Waalkens}, and H.~Broer.
\newblock {Cotangent bundle reduction and Poincar{\'e}-Birkhoff normal forms}.
\newblock {\em Preprint}.

\end{thebibliography}
%%%%%%%%%%%%%%
\end{document}